\documentclass[oupdraft]{bio}
\usepackage{url}
\usepackage{amsmath}
\usepackage{amsfonts}
\usepackage{bbm}
\usepackage{graphicx,psfrag,epsf}
\usepackage{enumerate}
\usepackage{makecell}
\usepackage{adjustbox}

\history{Received -;
revised -;
accepted -}

\begin{document}

\title{A flexible Bayesian framework for individualized inference via adaptive borrowing}

\author{ZIYU JI$^\ast$, JULIAN WOLFSON\\[4pt]
\textit{Division of Biostatistics, School of Public Health, University of Minnesota, \\
420 Delaware St.SE, Minneapolis, Minnesota 55455, U.S.A.}
\\[2pt]
{jixxx311@umn.edu}}

\markboth%
{Z. Ji and others}
{A flexible Bayesian framework for individualized inference}

\maketitle

\footnotetext{To whom correspondence should be addressed.}

\begin{abstract}
{The explosion in high-resolution data capture technologies in health has increased interest in making inferences about individual-level parameters. While technology may provide substantial data on a single individual, how best to use multisource population data to improve individualized inference remains an open research question. One possible approach, the multisource exchangeability model (MEM), is a Bayesian method for integrating data from supplementary sources into the analysis of a primary source. MEM was originally developed to improve inference for a single study by asymmetrically borrowing information from a set of similar previous studies and was further developed to apply a more computationally intensive symmetric borrowing in the context of basket trial; however, even for asymmetric borrowing, its computational burden grows exponentially with the number of supplementary sources, making it unsuitable for applications where hundreds or thousands of supplementary sources (i.e., individuals) could contribute to inference on a given individual. In this paper, we propose the data-driven MEM (dMEM), a two-stage approach that includes both source selection and clustering to enable the inclusion of an arbitrary number of sources to contribute to individualized inference in a computationally tractable and data-efficient way. We illustrate the application of dMEM to individual-level human behavior and mental well-being data collected via smartphones, where our approach increases individual-level estimation precision by 84\% compared with a standard no-borrowing method and outperforms recently-proposed competing methods in 80\% of individuals.} {Bayesian model averaging; individualized inference; multisource data borrowing; supplementary data.}
\end{abstract}

\section{Introduction}
\label{sec:intro}

Over the past few decades, the need to efficiently leverage large amounts of data from multiple sources to support analysis and decision-making has grown dramatically in many disciplines including the biomedical, social, and computational sciences \citep{raghupathi2014big,analytics2016age,matheny2020artificial}. Specifically in the biomedical area, researchers have developed techniques for integrating supplemental or historical data into the analysis of a primary clinical trial in order to increase the precision of parameter estimation \citep{viele2014use, han2017covariate}. Data borrowing can also be used when the primary source consists of multiple measurements on a single individual and the target is individual-level inference; for example, \cite{mejia2015improving} improves the reliability of subject-level resting-state fMRI parcellation by borrowing strength from a larger population of subjects, and \cite{jonsen2016joint} models animal movement behaviours by borrowing behavioural states of other individuals. In multi-source borrowing problems, the key challenge is determining how much to borrow from each supplementary source. Many possible approaches to optimal data borrowing have been proposed, with most belonging to two general perspectives: static shrinkage estimators with specified amounts of borrowing \citep{pocock1976combination,whitehead2008bayesian,hobbs2011hierarchical,rietbergen2011incorporation,french2012using} and adaptive weighting on supplementary sources \citep{smith1995bayesian,neuenschwander2010summarizing,doi2011meta,murray2014semiparametric,rover2020dynamically,papanikos2020bayesian}. 


The information borrowing process requires careful consideration of between-source heterogeneity; for example, conventional statistical methods which assume fully exchangeable supplementary data sources are deficient because the results are sensitive to inter-cohort bias. As a simple example, in the clinical trial context, one may wish to borrow more strongly from historical studies conducted more recently and on more similar populations to the primary study of interest. One promising method in this area is the multisource exchangeability model (MEM) proposed by \cite{kaizer2018bayesian}, a Bayesian model averaging \citep{hoeting1999bayesian,fragoso2018bayesian} method which integrates data arising from a mixture of exchangeable and nonexchangeable supplementary sources into the analysis of a primary source. The MEM considers all possible subsets of $M$ supplementary sources and computes a weight for each subset (based on the exchangeability of the sources in that subset with the primary source) that is used in model averaging. Since it considers all subsets of the $M$ sources, a major limitation of the MEM approach is that it scales exponentially with the number of supplementary data sources, limiting its application to cases with a small number of supplementary sources \citep{hobbs2018bayesian,kaizer2019basket}. With the ubiquity of high-resolution data capture technologies including wearable sensors and biomedical imaging, individual-level inference could be informed by data from hundreds or thousands of supplementary sources (i.e., individuals), rendering the standard MEM computationally intractable. Recently, \cite{brown2020iterated} introduced the iterated multisource exchangeability model (iMEM), which identifies, in linear computational time, the $q$ most exchangeable sources to include in a final MEM model. While iMEM allows MEM to be applied with many sources, the number of included sources $q$ is fixed (usually to $\approxeq 10$) to ensure that the final "all subsets" MEM remains computationally feasible. Fixing the number of included sources in this way may induce bias and/or inefficiency if the number of truly exchangeable sources is either smaller or (as is common in individual inference problems) much larger than $q$. 

A motivating example is from Daynamica \citep{fan2015smartrac}, a research oriented mobile phone application that records participants' daily activities and trips as well as their reported emotional status. There is evidence showing that the daily trip activities are associated with people's emotional states \citep{abou2011effect,de2016travel,zhu2018daily} hence transportation researchers are interested in developing tools to help promoting traffic safety and mental well-being through daily trips. However, the effects tremendously vary between individuals and the data collected from each person is usually not enough for obtaining good individualized statistical inferences on all combinations between trip modes and emotional states. To improve the estimating precision and have the tool available on mobile devices, the MEM framework could be applied but at the same time must to be able to automatically borrow from a large number of homogeneous or heterogeneous individuals (i.e., supplementary sources) without being overly computationally burdensome.

In this article, we propose the data-driven MEM (dMEM) approach, a novel method that is capable of identifying and leveraging many potentially nonexchangeable supplementary sources to improve inference on a primary source in a computationally tractable way. Inspired by MEM and iMEM, dMEM retains their desirable properties (avoiding the limiting assumption of exchangeability among the supplemental sources\citep{kaizer2018bayesian} and reducing the dimension of the MEM model space in the presence of many supplementary sources \citep{brown2020iterated}) while allowing an arbitrary number of supplementary sources to contribute to the final MEM model. Our simulation studies show that dMEM is able to leverage data from a much larger number of supplementary sources than MEM and iMEM, resulting in larger effective supplemental sample size, higher precision, better posterior efficiency and lower bias than its predecessors. 

The remainder of the article proceeds as follows. Section~\ref{sec:overview} provides an overview of MEM and iMEM, and introduces key notation. Section~\ref{sec:method} describe the dMEM methodology with Section~\ref{sec:justification} describes a motivating study showing how combining supplementary data sources can improve inference on a primary source. Section~\ref{sec:prop_asymp} describe the computational complexity and asymptotics of dMEM. Section~\ref{sec:sim} presents a simulation study comparing the performance of dMEM with competing methods across multiple scenarios. Section~\ref{sec:app} provides a real-world application of dMEM. We conclude with a brief discussion in Section~\ref{sec:diss}.

\section{Overview and Notation}
\label{sec:overview}


MEM was initially developed to address the problem of improving the estimation of treatment effects in a pivotal clinical trial by borrowing information from other potentially similar studies. It has been shown to achieve both less bias and greater efficiency when compared with competing methods \citep{kaizer2018bayesian}. Suppose we have a single primary source $P$ with $n$ observations and $H$ independent supplementary sources with $n_h$ observations each. A supplementary source is said to be \emph{exchangeable} with the primary source when they share the same value of the parameter of interest. All the sources together compose the observed data $D$. Our goal is to estimate the parameter $\theta_p$ for the primary source. Corresponding parameters for supplementary source $h$ are denoted as $\theta_h$; sources $h$ and $p$ are exchangeable if $\theta_h=\theta_p$, in which case we define an exchangeability indicator $S_h=1$. 

Under the standard MEM framework, a model $\Omega_k, k=1,\dotsc,K$ is associated with each of the $K=2^H$ possible exchangeability configurations, distinguished by a set of source-specific exchangeability indicators ($S_1=s_{1,k},\dotsc,S_H=s_{H,k}$) where $s_{h,k} \in \{0,1\}$. Bayesian model averaging is applied to average across all the $K=2^H$ possible models representing different exchangeability scenarios. Posterior weights $w_k$ are estimated for each possible model $\Omega_k$ and used for posterior inference on $\theta_p$. 

Let $\Theta = (\theta_p,\theta_1,\dotsc,\theta_h)$, $L(\Theta \mid D,\Omega_k)$ be the likelihood for data $D$ under model $\Omega_k$, and denote by $\pi(\Theta \mid \Omega_k)$ the prior density of $\Theta$ under $\Omega_k$. Then, the integrated marginal likelihood of $\Omega_k$ is $$p(D \mid \Omega_k)=\int L(\Theta \mid D,\Omega_k)\pi(\Theta \mid \Omega_k)d\Theta.$$ Posterior weights for each model $\Omega_k$ are $$w_k=p(\Omega_k \mid D)=\frac{p(D \mid \Omega_k)\pi(\Omega_k)}{\sum^{K}_{i=1}p(D \mid \Omega_i)\pi(\Omega_i)},$$ where $\pi(\Omega_k)$ is the prior probability that $\Omega_k$ is the true model, and is independently specified for each model with respect to sources, as $\pi(\Omega_k)=\pi(S_1=s_{1,k},\dotsc,S_H=s_{H,k})=\pi(S_1=s_{1,k})\times\dotsc\times\pi(S_H=s_{h,k}).$ The posterior distribution of $\theta_p$ given data $D$ is the weighted average of the $K$ MEM posteriors: $$p(\theta_p \mid D)=\sum^{K}_{k=1}w_k p(\theta_p \mid \Omega_k,D).$$ 

The Gaussian special case of MEMs has a number of desirable properties. It can be easily shown that with known variance, a non-informative prior on all $H$ sources  $\left(\pi(S_h=1)=\frac{1}{2}\right)$, and a flat prior on the Gaussian mean $\mu_k$ \citep{gelman2006prior} for each model ($\pi(\mu_k \mid \Omega_k) \propto 1$), there are closed forms for $p(D \mid \Omega_k)$, posterior weights $w_k$, and the posterior distribution of $\mu$, which is a Gaussian mixture with model-specific posteriors $p(\mu_p \mid \Omega_k,D)$ weighted by the posterior weights $w_k$. \citep{kaizer2018bayesian}



While MEM has appealing theoretical properties and practical benefits, it scales exponentially with the number of secondary sources. In some circumstances there may be dozens to thousands of secondary sources available, the required running time and computational power make applying MEM not feasible. In response to this computational limitation, \citeauthor{brown2020iterated} proposed the iterated MEM (iMEM). iMEM is based on the principle that marginal posterior MEM weights accurately indicate which supplementary sources are most similar to the primary source. In iMEM, marginal MEM models $M_h$ are fit considering only the primary source and one supplementary source $h$. The posterior weight $w^*_h$ from the marginal model $M_h$ serves as the similarity score $T_h$ between supplementary source $h$ with the primary source. The final step is applying the MEM framework on the primary source with the top $q$ supplementary sources with highest $T_h$ to obtain inference on $\theta_p$. The number of sources chosen for the final MEM, $q$, is determined by the available computational resources; on a standard desktop machine or eight regular cluster cores, $q$ must generally be smaller than 15 to make the final MEM computationally feasible with current computing technology.



Although iMEM provides an inspiring way to handle large number of supplementary sources, there are few reasons showing iMEM is not a perfect solution. First, the number of truly exchangeable sources is generally unknown in real-world data, and there is no obvious and straightforward way to choose $q$. \citeauthor{brown2020iterated} showed that setting $q$ to the maximum possible number of sources that can be accommodated computationally in the final MEM is not always the best strategy if the number of truly exchangeable sources is less than that value. The estimating efficiency may be reduced and the wrongly included nonexchangeable sources may also introduce extra bias to the posterior estimations. Second, in situations with many exchangeable secondary sources, including at most $q$ sources risks omitting other truly exchangeable secondary sources and "wastes" exchangeable data which could be used to improve our parameter inference. 


\section{dMEM: Data-driven Multisource Exchangeability Models}
\label{sec:method}

Our goal is this paper, therefore, is to develop a data-driven method for identifying and including all exchangeable secondary sources in a final MEM while limiting computational complexity. We propose the data-driven multisource exchangeability models (dMEM) method, a two-stage approach which selects and aggregates secondary sources for inclusion in a full MEM. Figure~\ref{fig:diagram} presents a diagram of the method. 


\subsection{Proof of Concept}
\label{sec:justification}

In this section, we present a simulation example to illustrate the potential gains achieved by incorporating data from additional supplementary sources beyond the "top $q$" sources selected by iMEM. Figure~\ref{fig:motivsim} compares the posterior mean-squared error (MSE) and expected supplemental sample size (ESSS, proposed by \citet{hobbs2013adaptive}, and defined in Section \ref{sec:sim}) of dMEM and iMEM with $q = 10$ in a simple scenario where there are 100 supplementary sources of which 60 are exchangeable with the primary source. dMEM achieves much lower MSE with about 4 times higher ESSS compared with iMEM. The results indicate that it is possible to outperform iMEM on both estimation accuracy and efficiency when the limitation on computational power exists. In addition, among all the combining methods we have tested, the option with 10 clusters and 6 exchangeable groups in a cluster achieves better performance than all the other methods, which suggests we want to include only the exchangeable sources in the final MEM.

\subsection{Source Selection}
\label{sec:method_step1}

In the real world, we cannot know which supplementary sources are truly exchangeable with the primary source. However, as shown in the results of the proof of concept simulation, it is beneficial to select only exchangeable sources for inclusion in the final MEM model.  \citeauthor{brown2020iterated} showed that the marginal posterior weight for each source ('marginal score' or 'marginal weight') could be used as a quantitative measure of its exchangeability; thus, we propose to select for inclusion in the final MEM supplementary sources whose marginal scores exceed a data-driven threshold. 

Suppose we have calculated marginal scores $\{w^{*}_{1},\dotsc,w^{*}_{H}\}$ for $H$ supplementary sources with mixed exchangeability, sorted as $\{w^{*}_{i_1},\dotsc,w^{*}_{i_H}\}$ (for simplicity of exposition, we assume that $w^*_j \neq w^*_k$ for $j \neq k$, i.e., there are no tied scores). In trivial cases, either all the supplementary sources are selected and dropped, which will be discussed at the end of this section. Otherwise, a data-driven threshold $\tau$ must lie between two consecutive weights $w^{*}_{i_c}$ and $w^{*}_{i_{c+1}}$ and split the sources into two sets $\{w^{*}_{i_1},\dotsc,w^{*}_{i_c}\}$ and $\{w^{*}_{i_{c+1}},\dotsc,w^{*}_{i_H}\}$. The sources $S_{\geq \tau} = \{S_{i_1},\dotsc,S_{i_c}\}$ having weights greater than or equal to $\tau$, are classified as exchangeable; the sources $S_{< \tau} = \{S_{i_{c+1}},\dotsc,S_{i_H}\}$ are dropped. 

The question remaining is how the threshold $\tau$ should be determined. The asymptotic existence of a change-point in marginal weights is shown in Section~\ref{sec:asymp}, since the marginal weights of exchangeable and nonexchangeable sources approach 1 and 0 respectively as $n$ (the number of observations per source) goes to infinity. Hence, in finite samples we seek to identify a break-point in the marginal scores of sources. We propose to use a likelihood ratio test-based change-point detection method first proposed by \cite{hinkley1970inference}. The presence of a change-point at a fixed value $\tau_c$ is tested via a hypothesis test wherein the null hypothesis $H_{0}$ is no change-point in the mean marginal scores at $\tau_c$ and the alternative hypothesis $H_{1}$ corresponds to the occurrence of a change-point at $\tau_c$. Under $H_{1}$, for a model with change-point at $w^*_{i_c}$, the log likelihood is given as $$l(w^{*}_{i_c}) = log(f(w^{*}_{i_1},\dotsc,w^{*}_{i_c} \mid \hat{\theta_{1}})) + log(f(w^{*}_{i_{c+1}},\dotsc,w^{*}_{i_H} \mid \hat{\theta_{2}})),$$ where $f(\cdot \mid \theta)$ is the probability density function of the weights corresponding to parameters $\theta$ and $\hat{\theta}$ is the maximum likelihood estimator of the parameters. While under $H_{0}$, we have $$l(w^{*}_{0}) = log(f(w^{*}_{i_1},\dotsc,w^{*}_{i_H} \mid \hat{\theta_{0}})).$$ Then, the likelihood ratio test statistic could be represented as  $$\lambda = -2[l(w^{*}_{0})-\max_{w^{*}_{i_c}}l(w^{*}_{i_c})].$$ The null hypothesis is rejected when $\lambda > b$, where $b$ is a pre-determined threshold. A variety of methods have been proposed to get an appropriate value of $b$, but since this is not the main topic of this paper, we do not pursue this detail here and rely on the rules implemented in software. The estimated location of the change-point is decided by the $\hat{w^{*}_{i_c}}$ that maximizes $l(w^{*}_{i_c})$. The \texttt{changepoint} package (version 2.2.2) in R \citep{killick2014changepoint} is used to apply the likelihood ratio test-based single change-point detection in our method.

A potential concern is that the likelihood ratio test-based change-point detector uses a Gaussian likelihood, while there is no guarantee that marginal scores are Gaussian. In auxiliary simulations (shown in Web Appendix A), we found that posterior weights for exchangeable and nonexchangeable sources were generally not normally distributed, nor did they seem to follow any particular parametric family. However, as demonstrated by our simulations in Section~\ref{sec:sim}, the performance of the change-point detector appears to be very good even though the underlying distributional assumption is wrong, likely because the likelihood ratio is primarily driven by a difference in means between sources with marginal scores above and below the change-point.

In Section~\ref{sec:asymp}, we prove that, asymptotically in non-degenerate cases, a change-point always exists in the set of marginal scores. In practice, there are two circumstances under which no change-point may be detected. When all sources have similar marginal weights, then we must decide whether to discard all or keep all the sources. A rule-of-thumb threshold of 'low marginal weights' is 0.2, depending on the sample size of both primary and supplementary source. Otherwise, the recommended approach is passing all the sources to the next step. Another reason the change-point may not be detected is that the ordered marginal scores for all sources decrease at a constant rate. In light of the asymptotic behavior of the weights, and as shown in Section~\ref{sec:sim_scen2}, this scenario is uncommon in practice since the change-point detector and marginal weights are sensitive to heterogeneity. In order to better handle this case, additional to the change-point detector, one may always choose to use a rule-of-thumb threshold or a hard threshold generated from the distribution of marginal posterior weights to eliminate some of the sources. No threshold is used in this paper for clearer demonstrations.

\subsection{Source Clustering}
\label{sec:method_step2}

The second step of dMEM is combining the remaining sources into larger clusters and running a final MEM on the set of "new" sources defined by these clusters. In the simulation study (see Section~\ref{sec:sim}), we compare different ways of clustering the selected sources. Although the choice of clustering methods does not influence the final result asymptotically, it can influence performance of dMEM with finite samples. The clustering of supplementary sources matters potentially when the change-point detector is not perfectly selecting the truly exchangeable sources. More guidance on source clustering is available in Section~\ref{sec:guide}.  

Finally, we fit a MEM on the combined clusters $S^{*}_1,\dots,S^{*}_M$ to obtain estimation and inference of our parameter of interest. The final results are generated assuming a Gaussian mixture MEM with known variances as stated in Section~\ref{sec:overview}. That is, the posterior mean is calculated as a finite mixture of Gaussian distributions with known variances.

\section{Properties and Asymptotics}
\label{sec:prop_asymp}

\subsection{Data-driven MEM Properties}
\label{sec:prop}


The computational complexity of MEM grows exponentially with an increasing number of supplementary sources. In MEM, there are $N_{MEM} = 2^{H}$ different models to estimate for a set of $H$ supplementary sources. dMEM keeps the computational simplicity of marginal iMEM, with $N_{dMEM} = H + 2^{M}$ models where $M$ is the (fixed) maximum number of groups that can be fit in a MEM under the computational limitation. $N_{dMEM}$ is superior to $N_{MEM}$ especially when $H \gg M$, since the computational requirement grows exponentially for $N_{MEM}$ with respect to $H$ but only grows linearly for $N_{dMEM}$. 

Similar to iMEM, the $H$ marginal models could be individually fitted, which makes dMEM plausible to further improve computing efficiency by using parallel computing. Although the $2^{M}$ models in the final MEM are serial, since MEM do not have closed form and requires posterior sampling in most of the cases, which is usually computationally expensive, this is still a remarkable advantage for dMEM especially when there are large numbers of supplementary sources. 

\subsection{Data-driven MEM Asymptotics}
\label{sec:asymp}

The bias-variance trade-off is heavily considered by statisticians when designing methods for supplemental information borrowing.  When all exchangeable supplementary sources are assigned weights 1 and all nonexchangeable sources have weights 0, MEM should achieve its lowest bias. \cite{kaizer2018bayesian} provides the theoretical basis of the MEM approach by showing the consistency of both model-specific posterior weights and the posterior mean for the true mean $\theta_p$ in the MEM framework. Here, we demonstrate the consistency of source-specific dMEM marginal posterior weights. In addition, the results also show the consistency of the exchangeability of clustered sources and the existence of a change-point in marginal posterior weights to perfectly separate exchangeable and nonexchangeable supplementary sources. 

The dMEM marginal posterior weights used in source selection are calculated under the marginal iMEM framework as shown in Section~\ref{sec:overview}. Theorem~\ref{thm:consist_marginal} is shown by \cite{brown2020iterated} as Theorem 4.3. This is a special case of the consistency of MEM posterior weights, where the models only have one supplementary source. Thus, marginal posterior weights could be treated as measures of exchangeability for supplementary sources, which further allow us to select exchangeable sources based on the marginal posterior weights.

\begin{theorem}[Consistency of dMEM Marginal Posterior Weights]
\label{thm:consist_marginal}
As $n,n_{1},\dotsc,n_{H} \rightarrow \infty$, for supplementary source $S_{h},h=1,\dotsc,H$, if $s_{h}=1$ then the dMEM marginal posterior weight $w^{*}_{h} \rightarrow 1$, otherwise $w^{*}_{h} \rightarrow 0$.
\end{theorem} 

The result shown in Theorem~\ref{thm:cp_existence} significantly supports the source separation procedure and ensures that we can asymptotically select all the truly exchangeable sources into the final model. Theorem~\ref{thm:consist_exchange} confirms the consistency of exchangeability for clustered sources. It justifies that the clustered sources are asymptotically exchangeable with the true mean as long as the selected supplementary sources are all asymptotically exchangeable with the primary source. Also, the clustering method will not impact the exchangeability of the clustered sources when the sample sizes of supplementary sources approach infinity. This makes the final estimation under the MEM framework asymptotically unbiased.

\begin{theorem}
[Existence of Change-point and Perfect Separation]
\label{thm:cp_existence}
As $n,n_{1},\dotsc,n_{H} \rightarrow \infty$, if $ 0 < \sum^{H}_{h=1}s_{h} < H $, sort marginal posterior weights $\{w^{*}_{1},\dotsc,w^{*}_{H}\}$ in descending order $\{w^{*}_{i_1},\dotsc,w^{*}_{i_H}\}$, where $i_h (h=1,\dotsc,H)$ are ordering indexes. $\exists i_c$ s.t. $\mathbb{P}(i_t \leq i_c \mid s_{i_t}=1) = 1$ and $\mathbb{P}(i_t > i_c \mid s_{i_t}=0) = 1$ for $\forall t = 1,\dotsc,H$.
\end{theorem}

\begin{proof}
By Theorem~\ref{thm:consist_marginal}, under asymptotic assumptions, for marginal posterior weights $\{w^{*}_{1},\dotsc,w^{*}_{H}\}$ we have $\mathbb{P}(w^{*}_{h} = 0 \mid s_{h}=0) = 1$ and $\mathbb{P}(w^{*}_{h} = 1 \mid s_{h}=1) = 1$. \\
So, if $\{0,1\} \in {s_{1},\dotsc,s_{h}}$, for descending $\{w^{*}_{i_1},\dotsc,w^{*}_{i_H}\}$, $\exists i_c$ s.t. $w^{*}_{i_1},\dotsc,w^{*}_{i_c} \rightarrow 1$ and $w^{*}_{i_{c+1}},\dotsc,w^{*}_{i_H} \rightarrow 0$. \\
Thus, $\mathbb{P}(w^{*}_{i_t} \in \{w^{*}_{i_1},\dotsc,w^{*}_{i_c}\} \mid s_{i_t}=1) = \mathbb{P}(i_t \leq i_c \mid s_{i_t}=1) = 1$ and $\mathbb{P}(w^{*}_{i_t} \in \{w^{*}_{i_{c+1}},\dotsc,w^{*}_{i_H}\} \mid s_{i_t}=0) = \mathbb{P}(i_t > i_c \mid s_{i_t}=0) = 1$
\end{proof} 

\begin{theorem}[Consistency of Exchangeability for Clustered Sources]
\label{thm:consist_exchange}
As $n,n_{1},\dotsc,n_{H} \rightarrow \infty$, if $\mu_{t_1},\dotsc,\mu_{t_M}\rightarrow\mu$, $\forall t_m \in {1,\dotsc,H}$, then $\mu^*_{m}\rightarrow\mu$, where $\mu^*_{m}$ is the mean of the cluster $S^*_{m}$, which is combined by sources $S_{t_1},...,S_{t_M}$.
\end{theorem}

\begin{proof}
It is easy to show that $\mu^*_{M}=\frac{1}{M}\sum^M_{m=1}{\mu_{t_m}}\rightarrow\mu$.
\end{proof} 

\section{Simulation Study}
\label{sec:sim}

The goal of our simulation study is to compare the results between dMEM and existing methods. We also tested different ways of clustering the sources selected as potentially exchangeable in the first stage of dMEM. Details of tested methods are shown in Table~\ref{tab:simclstmeth}. Note that it is completely infeasible to run a full MEM under our simulating scenarios (averaging across $2^{100}$ or $2^{500}$ models), so marginal iMEM ($q=5$ or $10$) serves as the main comparator in our study (in \citet{brown2020iterated}, marginal iMEM is shown to outperform several other methods). In addition to iMEM, we also consider a K-means algorithm \citep{hartigan1979algorithm} as a comparator that assigns sources into 5 or 10 clusters by clustering the sample means of sources. Posterior variances, biases, root-mean-square errors (RMSE) and effective supplemental sample sizes (ESSS) are used as the quantitative measures to evaluate the methods. 

ESSS \citep{hobbs2013adaptive} is an extension of prior effective sample size \citep{morita2008determining}, defined as $ESSS_{MEM}=\sum^{K}_{k=1}\{w_k\frac{1/v+\sum^{H}_{h=1}s_{h,k}/v_h}{1/v}-1\}$, where $1/v$ is the posterior precision of the reference model with no borrowing, and $1/v+\sum^{H}_{h=1}s_{h,k}/v_h$ is the posterior precision for $\Omega_k$. A higher value of ESSS indicates more borrowing from supplementary sources. Under most of our tested scenarios, the ESSS of dMEM is usually 2-3 times the ESSS of iMEM and K-means. From the nature of the methods, it is expected that dMEM will utilize more data, hence have larger ESSS. So, ESSS will not serve as a main measure in our analysis; interested readers can refer to supplementary materials to obtain summaries of ESSS (Web Appendix B, Figure 4). 

For all scenarios and methods, we repeat each simulation 1000 times with different random seeds to obtain the final results. The exchangeability status is unknown in our simulation. Simulated results are shown in Figure~\ref{fig:sim_plot}. Additional simulations on the robustness of dMEM with non-Gaussian sources or small number of exchangeable sources are available as supplementary simulation scenarios in Web Appendix C. All calculations are completed with R version 4.0.2 \citep{R}.

\subsection{Scenario I: Nonexchangeable sources with the same mean}
\label{sec:sim_scen1}

Here, we show the performance of dMEM under the basic setting with a fixed number of nonexchangeable sources with common shared mean. For data generation, the primary source contains 20 observations that are randomly sampled from $N(0,1)$. There are 60 truly exchangeable supplementary sources and 40 nonexchangeable ones (with mean 1), which have sample sizes uniformly distributed between 15 and 25, with standard deviations generated from $\text{Unif}(0.5,1.5)$. The settings are intentionally designed to challenge the change-point detector, by limiting the exchangeability of nonexchangeable sources to be not significantly different from exchangeable sources. However, with $q=5$ or $10$ for iMEM, the probability of iMEM wrongly including nonexchangeable sources is low. Results are shown in sub-panel (a) of Figure~\ref{fig:sim_plot}.


All the clustering methods for dMEM perform better than iMEM. Evenly distributed, random clustering and random averaging method have the lowest RMSE, whose medians are almost 3 times smaller than the median RMSE of the best baseline method (iMEM $q=10$). Clustering methods with fewer clusters and involving single-source clusters have relatively higher RMSE. The posterior variance, bias and RMSE all increase with a smaller number of clusters. For the method 'single half \& combine half', although the mean of posterior variance is smaller, the range in bias is much larger than for random clustering, resulting in a higher RMSE. Overall, the trend is that clustering methods using more and smaller clusters outperform methods with fewer and larger clusters. Note that with variant of supplementary standard deviations and sample sizes, the result of this scenario is a more realistic hence difficult version of the simulation study shown in Figure~\ref{fig:motivsim}, which creates a more complex structure in term of exchangeability and yields a less significant advantage in results. 

\subsection{Scenario II: Varying degrees of nonexchangeability}
\label{sec:sim_scen2}

This scenario evaluates the performance of our method when some supplementary sources are "nearly" exchangeable with the primary source. The number of exchangeable supplementary sources is 50, plus 5 groups of 10 nonexchangeable sources with source means in each group set to $1,0.8,0.6,0.4$ and $0.2$, respectively. All the other settings are the same as in Scenario I. Supplementary sources with means closer to zero are even more difficult to identify as nonexchangeable based on marginal weights, and hence our changepoint-based method will likely inadvertently include some nonexchageable sources in the final MEM. Results are shown in sub-panel (b) of Figure~\ref{fig:sim_plot}.

As expected, due to the inclusion of some nonexchangeable sources in the final MEM, bias (and hence RMSE) are slightly larger in Scenario II than in Scenario I for all methods. However, the relative ranking of clustering methods is mostly unchanged. All dMEM methods with 10 clusters are equally well and better than other methods in term of RMSE. The median RMSE of dMEM methods is 1.6 times smaller than the median MSE of iMEM $q=10$. After looking at the marginal weights distribution of the supplementary sources, we found that the change-point detector struggled to identify sources as nonexchangeable when they had means below 0.4, which in this scenario corresponds to 0.4 standard deviations. Nonexchangeable sources most likely to be misclassified as exchangeable are those that are "nearly" exchangeable, and inclusion of these near-exchangeable sources in the final MEM drive up some bias, but this moderate misclassification of nonexchangeable supplementary sources does not lead to substantial inflation in RMSE. 

\subsection{Scenario III: Varying proportions of nonexchangeable sources}
\label{sec:sim_scen3}

This scenario shows the change in performance of dMEM with respect to the proportions of exchangeable sources when the total number of supplementary sources is fixed. The results of dMEM are compared with iMEM ($q=10$). There are 500 supplementary sources in total and we test the cases with $10\%, 20\%, 30\%, 40\%$ and $50\%$ exchangeable sources respectively for both dMEM and iMEM. All sources are otherwise sampled similarly as in scenario I. Results are shown in sub-panel (c) of Figure~\ref{fig:sim_plot}.

From the simulated RMSE, dMEM performs slightly worse than iMEM when the proportion is $10\%$, but increasingly achieves more advantages over iMEM when the proportion of exchangeable sources increases. The posterior variance of both iMEM and dMEM are decreasing with larger proportions, but the trend is more significant for dMEM. Meanwhile, the bias and RMSE dramatically decrease only for dMEM. For iMEM ($q=10$), although the true exchangeability level are kept same, the selected top 10 sources would have higher chance to have larger marginal weights and lead to smaller posterior variances when there are more exchangeable sources to choose from, but this effect does not help with bias and the general RMSE. Note that these results are only valid with the tested relative exchangeability levels of sources. When we increase the mean of nonexchangeable sources from 1 to 1.5, the trend in dMEM bias disappears and dMEM has lower posterior variance, bias and RMSE than iMEM under all circumstances.

 Along with another simulation on small number of exchangeable sources available in Web Appendix C Supplementary Simulation Scenario I, the only case that the dMEM may not outperform iMEM is when the exchangeability of sources are not clearly distinguishable and the number or proportion of exchangeable sources is very small (say, 5\% to 10\% of the total supplementary sources), the changepoint detector might not be able to drop enough sources and could result to higher bias. When the true number of exchangeable sources is greater than $q$, dMEM tends to borrow more from supplementary sources compared with iMEM, which explains the performance difference between iMEM and dMEM when the proportion of exchangeable source changes. In general, due to the potential bias-variance trade-off, dMEM has better performance compared with iMEM when a relatively higher proportion of sources are exchangeable.

\subsection{Implementation guidance}
\label{sec:guide}

From the simulating results shown above, all of the dMEM methods except '1 cluster' have lower posterior variance, bias and MSE on average with smaller ranges than iMEM under Scenarios I, II and III (cases with over 10\% exchangeable supplementary sources), and have approximately equal performance under Scenario III (10\% exchangeable supplementary sources) . The ESSS of dMEM is also 2 to 3 times higher than iMEM, which indicates dMEM is making a better use of supplementary sources. The Gaussian dMEM also demonstrates considerable robustness with non-Guassian supplementary sources, the corresponding evidence is available as the supplementary simulation scenario III in Web Appendix C. There are also cases when the marginal scores are not reflecting the correct exchangeability, an example is shown as the Supplementary Simulation Scenario II in Web Appendix C, where the performance of dMEM and iMEM is approximately the same. In conclusion, dMEM works well and robust under different scenarios, as well as improves the performance of iMEM with higher efficiency in data usage while keeps iMEM's advantage in computational complexity. 

Because the selected sources are all considered as potentially exchangeable, although evenly distributed clustering performs the best for most of the scenarios, random clustering also works not significantly worse than other tested clustering methods. An alternative to random clustering is taking the average of the estimates from multiple random clustering approaches to obtain the final results to smooth over the randomness of clustering. This option should be considered if there is a concern that the change-point detector may not work well or there is severe diversity in marginal posterior weights of selected supplementary sources. In addition, if eliminating bias is the priority, ordered combine is recommended. Also, if there is evidence that parts of the supplementary sources should be combined together, such as the affiliation to the same origin or a minor but clear change-points are detected, more case-specific decisions in the way of clustering have to be made. Nevertheless, random clustering is the default clustering method of dMEM. 

If it is expected to have very small proportion of supplementary samples to be exchangeable and bias becomes a major concern, we recommend either fixing the maximum number of selected sources at a smaller value or adding a penalty term related with the number of selected sources on the changepoint detector to limit the borrowing from supplementary sources.

\section{Application}
\label{sec:app}

In this section, we show a real-world application of dMEM using data collected by Daynamica \citep{fan2015smartrac}, a mobile phone application that automatically captures daily activities and trips using smartphone sensors and machine learning techniques. The data we use in this example comes from a study implemented in Minneapolis area; 356 participants were asked to use Daynamica for seven days and, for each activity and trip, rate the intensity of six emotions (happy, meaningful, stressful, tired, pain and sad) on a 1-7 scale. To provide individualized solutions on traffic safety and mental health, transportation experts are specifically interested in the person's emotional status during different modes of trips (walk, car, bus or bike), hence our interest is in estimating, for a given individual $i$, the mean emotional intensity $\mu^{i}(e,t)$ for emotion type $e$ and trip type $t$, for example $\mu^{ID=2021}(\text{stressful}, \text{car})$. Few study participants have over 20 trips for any individual trip mode, leading to low precision for estimating $\mu^{i}(e,t)$ if the individual-level mean  $\bar y^{i}(e,t) = \frac{1}{M^i(e,t)} \sum_j y^{i}_{j}(e,t)$ is used (here, $y^{i}_{j}(e,t)$ is the reported intensity of emotion $e$ from trip $j$ of type $t$ for subject $i$, and $M^i(e,t)$ is the total number of such trips). Thus, we may want to borrow trip- or activity-specific emotional intensity data from other individuals using MEM. However, due to the large number of potential supplementary sources, it is not possible to fit a full MEM.  The same dataset was analyzed in \cite{brown2020iterated} to illustrate the significant advantage of iMEM compared with naive approaches for estimating $\mu^{i}(e,t)$ such as calculating the simple mean from the individual or using the best linear unbiased predictor (BLUP) from a random effects model. Here, we explore when dMEM performs better or worse than iMEM. While the ordinal data in this application are not Gaussian, supplementary simulation scenario III  (Web Appendix C) shows that dMEM appears to perform relatively well with such data.

For this illustration, we estimated $\mu^{i}(e,t)$ for each individual $i$ and every $(e,t)$ combination for which $M^i(e,t) \geq 10$. Supplementary sources $j$ were included provided $M^j(e,t) \geq 5$. These criteria allowed estimation of 1910 $\mu^i(e,t)$ parameters. The methods we compare in this example are marginal iMEM with $q=10$ and dMEM with ordered combined clustering with $M=10$. The percent reduction in posterior standard deviation of dMEM and iMEM relative to the standard deviation of primary source is used to compare the performance of the methods. 

Overall, the percent reduction in posterior standard deviation by dMEM is 83.5\%, compared with 80.0\% for iMEM. The difference between dMEM and iMEM is relatively modest at 3.5\%, though 78.3\% of estimations yield a lower posterior SD with dMEM than iMEM and the ESSS of dMEM is over 2 times that of iMEM (1101 vs. 449); see Web Appendix B Figure 3. It is not entirely surprising that the overall mean value of pairwise difference in percent SD reduction is not big since both iMEM and dMEM perform well in most of the estimations: around 70\% of iMEM and 75\% of dMEM estimations reduce posterior SD by over 80\% compared to simple SD, so there is not much room for further improvement and the marginal gain from more intensive borrowing is modest in this specific data set.

Figure~\ref{fig:app_plot} provides some insight into situations where the performance of dMEM and iMEM differs substantially. Figure~\ref{fig:app_plot} (a) shows the number of selected sources of dMEM for each trip mode compared to the total number of available sources and, for reference, the constant number (10) selected by iMEM. For trip modes with a smaller number of available supplementary sources (bike and bus), dMEM typically borrows fewer than 10 supplementary sources, and yields a 10.7\% extra reduction in posterior standard deviation on average compared with iMEM. The top two panels of Figure~\ref{fig:app_plot} (b) illustrate two cases where, by applying the changepoint detector to the marginal MEM weights, dMEM (apparently correctly) identifies fewer supplementary sources for inclusion than iMEM. This shows that dMEM can be valuable as a mechanism for selecting the number of supplementary sources to include even if that number is quite small, addressing a concern raised in \citet{brown2020iterated}. The bottom two panels of Figure~\ref{fig:app_plot} (b) illustrate the opposite situation, where dMEM selects many more than 10 sources for inclusion in the final MEM. In both estimations, dMEM achieves a $>25\%$ reduction in posterior SD relative to iMEM.

Another situation where dMEM appears to perform well is in those cases where iMEM offers little performance benefit over the simple mean. If we define an SD reduction of less than 20\% relative to the simple mean as 'poor performance', there are 93 iMEM and 68 dMEM estimations perform poorly. Among those poor-performing iMEM estimations, the average difference in SD reduction between dMEM and iMEM is 24.1\% with dMEM outperforming iMEM in 87 (93.6\% of) cases. In contrast, among the 68 poor-performing dMEM estimations, the average difference in SD reduction between dMEM and iMEM is -13.6\%, and in only 7 out of 68 (10.3\% of) cases does iMEM outperform dMEM.

\section{Discussion}
\label{sec:diss}

dMEM is a novel method which can be used to borrow information from a large number of supplementary sources with the goal of improving inference about parameters in a primary source.  dMEM increases efficiency by including as many exchangeable sources as possible and minimizes bias by eliminating nonexchangeable sources. In contrast to existing methods whose computation time scales exponentially with the number of supplementary sources,  computations for dMEM grow linearly making it computationally feasible for much larger problems. With the massive increase in the number of distinct data sources available to inform estimation in biomedical research, dMEM is a potentially valuable tool for aggregating information from these sources. For example, dMEM could be applied to leverage information from electronic health databases coming from many different health systems or hospitals to inform estimation in a designed trial or observational study.

There are some similarities between Gaussian clustering problem and the problem that methods in MEM family could solve. The goal of MEM is to provide better inference for one or more specific primary sources by borrowing from the available data, while the objective of Gaussian clustering is to identify groups of observations with the same mean. There would be more similarities between dMEM and Gaussian clustering since the source selection procedure of dMEM is trying to detect the sources with higher probability of having the same mean with the primary source. Interested readers could explore more possibilities of merging the two ideas.

We note a few limitations and directions for future research.  First, as with MEM and iMEM, dMEM is only designed for estimating the means of the sources; while the univariate Gaussian case has closed-form posteriors, posterior sampling methods are required for sources with other distributions. How to generalize multisource exchangeability models to estimate more parameters and more complex models is an open research question. \cite{kotalik2020dynamic} and \cite{boatman2020borrowing} have built up extensions of MEM on estimating causal treatment effect as regression coefficients. Second, as shown in Web Appendix C Scenario II, the ordering of the marginal posterior weights does not always represent the ordering in exchangeability of supplementary sources. This issue exists in both iMEM and dMEM. The marginal posterior weight is a measure of apparent similarity that balances between the mean and the uncertainty, and hence it is possible to construct extreme scenarios where the apparent similarity does not reflect true exchangeability. Other measures of exchangeability could be further explored in the future. Third, we do not fully take the uncertainty in clustering into account. Our simulations showed that reusing the information of marginal weights in clustering does not provide too much benefit compared with random clustering. Other clustering methods such as unsupervised techniques could be used to potentially define clusters that yield improved inference in the final MEM. Lastly, although we found that a standard Gaussian likelihood ratio-based change-point detector performed well for separating exchangeable and nonexchangeable sources based on the posterior marginal weights, these weights are not normally distributed, and hence other change-point detection methods might yield improved performance.  

\section{Software}
\label{sec:software}

Software in the form of R code, together with the simulated results and application data set, is available on GitHub (https://github.com/Ziyu-Ji/Data-driven-Multisource-Exchangeability-Models).

\section{Supplementary Material}
\label{sec:supple}

Supplementary material is available online at
\url{http://biostatistics.oxfordjournals.org}.

\bibliographystyle{biorefs}
\bibliography{refs}

\begin{figure}[!p]
\begin{center}
\includegraphics[width=5.5in]{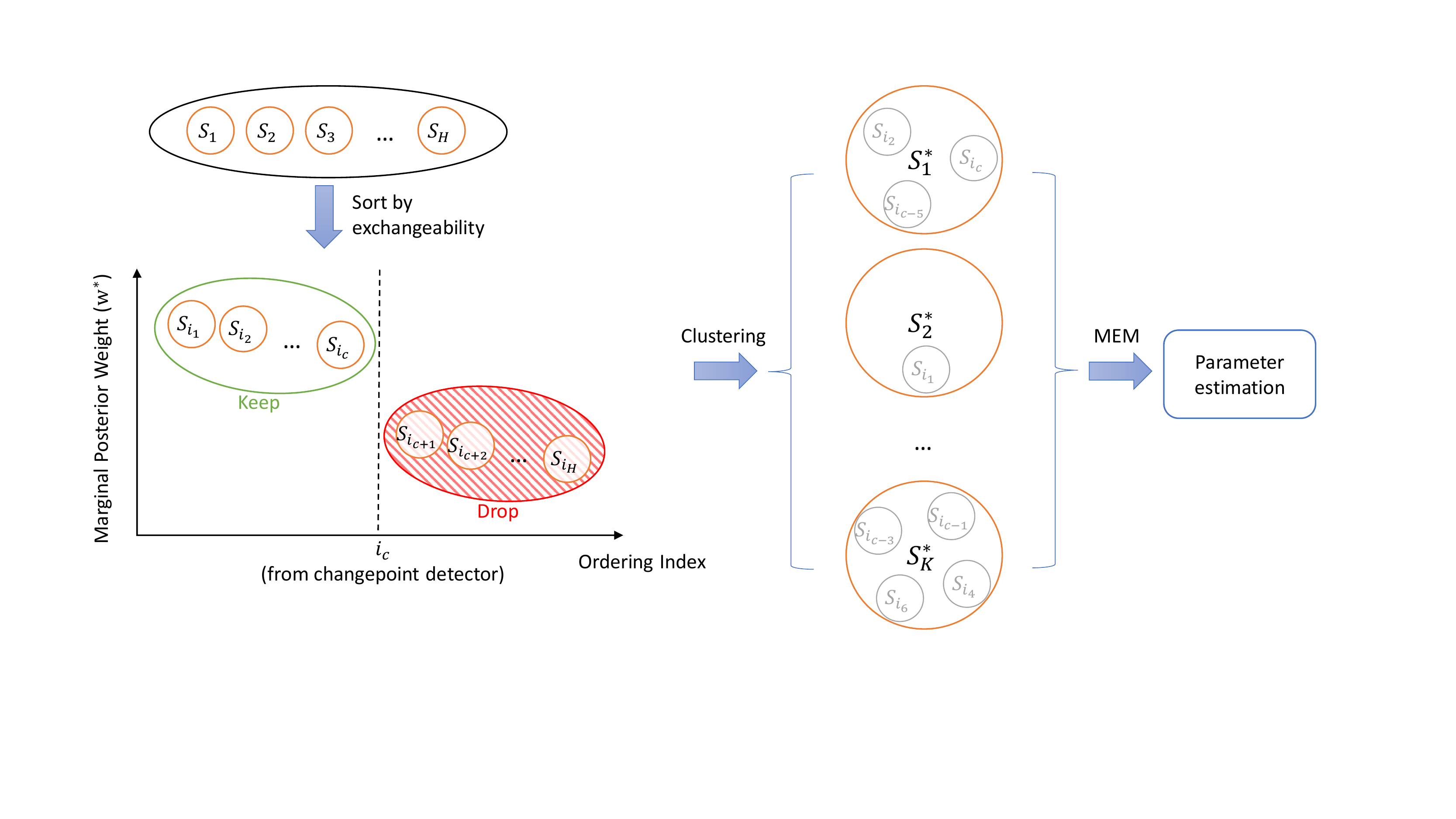}
\caption{Diagram of dMEM method \label{fig:diagram}}
\end{center}
\end{figure}

\begin{figure}[!p]
\begin{center}
\includegraphics[width=4.75in]{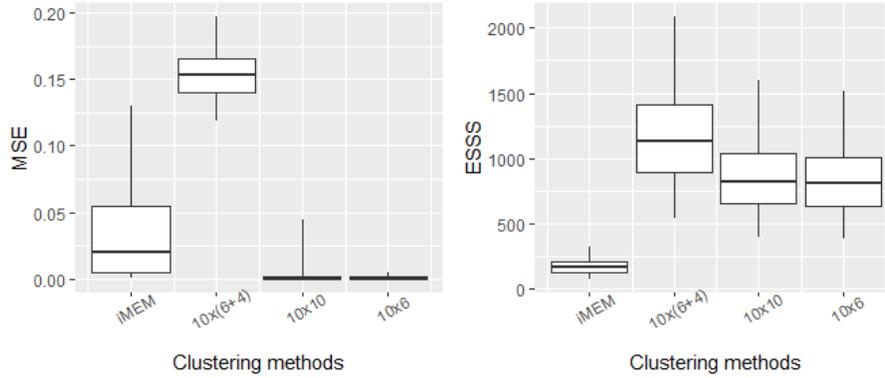}
\caption{Box plot of the posterior MSE and ESSS for proof of concept simulation study. The primary source follows $N(0,1)$. The supplementary sources include 60 exchangeable sources following $N(0,1)$ and 40 nonexchangeable sources following $N(1,1)$. All sources have sample size 20. The simulation is repeated 1,000 times to achieve the final results. The methods are described as follows. 'iMEM' means iMEM with 10 exchangeable sources. '$10 \times 10$' has 6 clusters randomly including 10 exchangeable sources and 4 clusters randomly including 10 nonexchangeable sources. '$10 \times (6+4)$' has 10 clusters, each cluster randomly includes 6 exchangeable sources and 4 nonexchangeable sources.
'$10 \times 6$' also has 10 clusters, with each cluster randomly includes 6 exchangeable sources.  \label{fig:motivsim}}
\end{center}
\end{figure}

\begin{table}[!p]
\begin{center}
\begin{adjustbox}{max width =\textwidth}
\begin{tabular}{l|cccp{11cm}}
\textbf{\makecell{Name of the \\ Methods}} & \textbf{ \makecell{Single-source\\ Clusters}} & \textbf{\makecell{Combined\\ Clusters}} & \textbf{\makecell{Randomly\\ Clustered}} & \textbf{Description}\\\hline\hline
kmeans & 0 & 10 & No & Baseline: assign all sources into 10 clusters by using K-means on their sample means. \\\hline
iMEM & 10 & 0 & No & Baseline: iMEM with final group size equal to 10. \\\hline
ordered combine & 0 & 10 & No & 10-fold split the descending ordered sources, combine the 10 sources in the first fold with the highest weights as the first cluster, then combine the 10 sources in the second fold as the second cluster and etc. \\\hline
evenly distributed & 0 & 10 & No & 10-fold split the sorted sources, combine the sources with highest weights in each fold as a cluster, then combine the sources with second highest weights in each fold and etc. \\\hline
single half \&  combine half & 5 & 5 & No & Use top 5 sources with the highest weights as single-source clusters and combine the rest sources in order into 5 clusters. \\\hline
random & 0 & 10 & Yes & Randomly combine the sources into 10 clusters. \\\hline
random averaging & 0 & 10 & Yes & Randomly combine the sources into 10 clusters for 10 times and take the average of their posterior results. \\\hline
3 cluster & 0 & 3 & Yes & Randomly combine the sources into 3 clusters. \\\hline
1 cluster & 0 & 1 & Yes & Combine all sources. \\
\end{tabular}
\end{adjustbox}
\end{center}
\caption{Tested clustering methods in simulation study \label{tab:simclstmeth}}
\end{table}

\begin{figure}[!p]
\begin{center}
\includegraphics[width=6in]{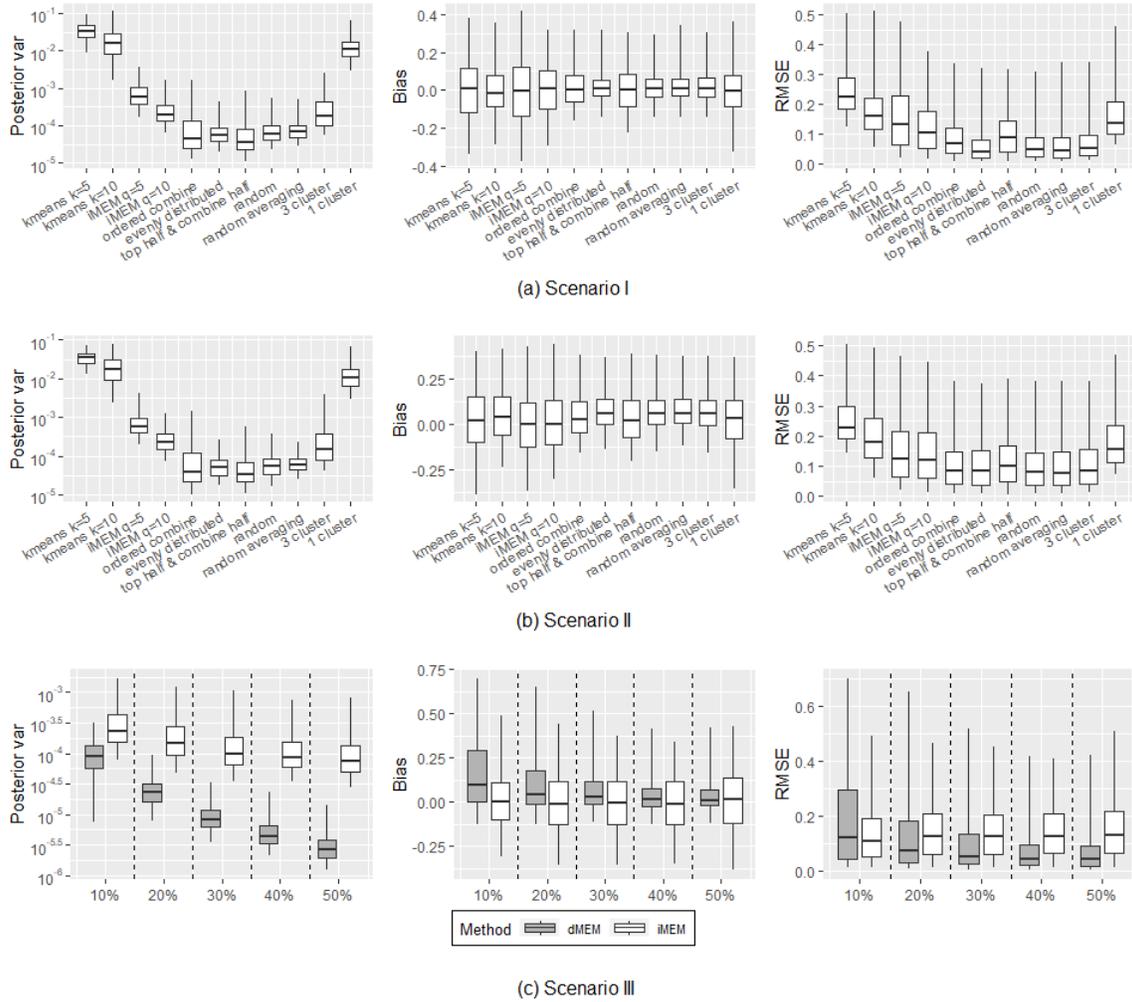}
\caption{Log posterior variance, bias and RMSE (root-square-mean error) for simulated scenarios. The box plots are showing the 2.5\%, 25\%, 50\%, 75\% and 97.5\% percentiles of the 1000 simulating results for each method. In Scenario III, the x-axis is the proportion of truly exchangeable sources to the total 500 supplementary sources. \label{fig:sim_plot}}
\end{center}
\end{figure}

\begin{figure}[!p]
\begin{center}
\includegraphics[width=6in]{app_plot.pdf}
\caption{(a) Bar blot for the numbers of selected supplementary sources of iMEM (always be 10 in our case) or dMEM, and the means of total available supplementary sources. The error bars are the 5\% and 95\% percentile of selected sources by dMEM. \protect\linebreak
(b) Examples of sorted marginal weights and selected sources. The vertical lines are the position of the last selected source by iMEM and dMEM. The top two examples have small numbers of total supplementary sources and the final estimations benefit from selecting less sources. The bottom two examples with large numbers of total supplementary sources benefit from selecting more sources. Note that there are examples with large numbers of total supplementary sources would benefit from selecting less sources, but less common. \label{fig:app_plot}}
\end{center}
\end{figure}

\end{document}